\documentclass[letterpaper,11pt]{article}
\usepackage{amsthm}
\usepackage{amsfonts}
\usepackage{amsmath}
\usepackage{amssymb}
\usepackage{graphicx}
\usepackage{color}
\usepackage[colorlinks=true,citecolor=black,linkcolor=black,urlcolor=blue]{hyperref}
\usepackage{authblk}

\newtheorem{theorem}{Theorem}
\newtheorem*{theorem*}{Theorem}
\newtheorem{conjecture}{Conjecture}
\newtheorem{question}{Question}

\newtheorem{lemma}[theorem]{Lemma}

\newtheorem{proposition}[theorem]{Proposition}

\newcommand{\fs}{\mathcal S}

\newcommand{\fc}{\mathcal C}

\newcommand{\p}{{\mathbf{Pr}}}

\newcommand{\E}{\mathbb{E}}

\newcommand{\gall}{\gamma_{\ell\hspace{-.058em}\ell}}

\title{A superlocal version of Reed's Conjecture}
\author{Katherine Edwards\thanks{Email: ke@princeton.edu. Supported by an NSERC PGS-D Fellowship and a Gordon Wu Fellowship.}}
\affil{Department of Computer Science\\Princeton University, Princeton, NJ}
\author{Andrew D.~King\thanks{Email: andrew.d.king@gmail.com.  This research was supported in part by an EBCO/Ebbich Postdoctoral Scholarship and the NSERC Discovery Grants of Pavol Hell and Bojan Mohar.}}
\affil{D-Wave Systems, Burnaby, BC}
\begin{document}

\maketitle

\begin{abstract}
Reed's well-known $\omega$, $\Delta$, $\chi$ conjecture proposes that every graph satisfies $\chi \leq \lceil \frac 12(\Delta+1+\omega)\rceil$.  The second author formulated a {\em local strengthening} of this conjecture that considers a bound supplied by the neighbourhood of a single vertex.  Following the idea that the chromatic number cannot be greatly affected by any particular stable set of vertices, we propose a further strengthening that considers a bound supplied by the neighbourhoods of two adjacent vertices.  We provide some fundamental evidence in support, namely that the stronger bound holds in the fractional relaxation and holds for both quasi-line graphs and graphs with stability number two.  We also conjecture that in the fractional version, we can push the locality even further.
\end{abstract}

\section{Introduction}

We consider simple graphs with clique number $\omega$, maximum degree $\Delta$, chromatic number $\chi$, and fractional chromatic number $\chi_f$ (we will define $\chi_f$ later).  For a graph $G$ and a set of vertices $S$ we denote the subgraph of $G$ induced by $S$ by $G|S$.  For a vertex $v$ we use $N(v)$ and $\tilde N(v)$ to denote the neighbourhood and closed neighbourhood of $v$, respectively.  We use $\omega(v)$ to denote $\omega(G|\tilde N(v))$, i.e.\ the size of the largest clique containing $v$.  When the graph in question is not clear, we specify with a subscript, for example $N_G(v)$.

The work in this paper revolves around Reed's $\omega$, $\Delta$, $\chi$ conjecture \cite{reed98}, which itself can be broadly considered as a generalization of Brooks' Theorem.  Brooks' Theorem states that whenever $\Delta \geq 3$, a graph with maximum degree $\Delta$ is $\Delta$-colourable unless it has the obvious obstruction: a clique of size $\Delta+1$.  Reed's Conjecture is much more general:

\begin{conjecture}[Reed's Conjecture]\label{con:1}
Every graph satisfies $\chi \leq \lceil \frac 12(\Delta+1+\omega) \rceil$.
\end{conjecture}

In other words, a graph with maximum degree $\Delta$ is $(\Delta+1-k)$-colourable unless it contains a clique of size at least $\Delta+2-2k$.

This conjecture is known to hold for claw-free graphs \cite{kingthesis} and some other hereditary families of graphs \cite{aravindks11}.  Furthermore Reed proved that the fractional relaxation holds, even without the round-up -- a proof appears in \cite{molloyrbook}:

\begin{theorem}[Fractional relaxation]\label{thm:frac1}
Every graph satisfies $\chi_f \leq \frac 12(\Delta+1+\omega)$.
\end{theorem}

For a graph $G$ we let $\gamma(G)$ and $\gamma'(G)$ denote $\lceil \frac 12(\Delta+1+\omega) \rceil$ and $\frac 12(\Delta+1+\omega)$, respectively.  As observed by McDiarmid (Exercise 21.1 in \cite{molloyrbook}; a proof appears in Chapter 2 of \cite{kingthesis}), Theorem \ref{thm:frac1} can be strengthened so as to consider only the possible bounds achieved in the closed neighbourhood of a vertex. Letting $\gamma'_\ell(v)$ denote $\gamma'(G|\tilde N(v))$ and $\gamma'_\ell(G)$ denote $\max_{v\in V(G)}\gamma'_\ell(v)$, we have:

\begin{theorem}[Local fractional relaxation]\label{thm:frac2}
Every graph $G$ satisfies $\chi_f(G) \leq \gamma'_\ell(G)$.
\end{theorem}

Inspired by structural observations, the second author conjectured that this {\em local strengthening} holds in the integer setting \cite{kingthesis}. Let $\gamma_\ell(v) $ denote $ \gamma(G|\tilde N(v))$ and let $\gamma_\ell(G)$ denote $\max_{v\in V(G)}\gamma_\ell(v)$.

\begin{conjecture}[Local Reed's Conjecture]\label{con:2}
Every graph $G$ satisfies $\chi(G) \leq \gamma_\ell(G)$.
\end{conjecture}

A typical example of a graph $G$ for which $\gamma(G)$ is far from $\chi(G)$ is the star $K_{1,r}$.  For such graphs we have $\gamma_\ell(G) = \gamma(G)$, so the bound offered by the local conjecture isn't any better.  And yet a greedy colouring algorithm can very easily $2$-colour a star.  Furthermore, examples of the tightness of Reed's Conjecture tend to be vertex-transitive, or at least very nearly regular.  So can we get a better bound when vertices that are hard to colour (i.e.\ have high $\gamma_\ell(v)$) form a stable set?  The answer, at least in the fractional setting and for certain graph classes, is yes.

\subsection{The superlocal strengthening}

Our idea is that a graph should be easy to colour if no two vertices with high $\gamma_\ell(v)$ are adjacent.  This gives rise to the invariants $\gall$ and $\gall'$, which we define as follows:
\begin{align*}
\textrm{For }uv\in E(G), \textrm{ define } \gall '(uv) &\textrm{ as } \tfrac 14(d(u)+d(v)+\omega(u)+\omega(v)+2)\\ &\hspace{.25em} = \tfrac 12(\gamma'_\ell(u)+\gamma'_\ell(v)).\\
\textrm{Define } \gall '(G) &\textrm{ as }\ \max_{uv\in E(V)}{\gall '(uv).}\\
\textrm{For }uv\in E(G), \textrm{ define } \gall(uv) &\textrm{ as }\ \lceil\gall '(uv)\rceil.\\
\textrm{Define } \gall(G) &\textrm{ as }\ \lceil \gall '(G)\rceil.
\end{align*}

We pose the natural conjecture regarding these invariants:

\begin{conjecture}[Superlocal Reed's Conjecture]\label{con:3}
Every graph $G$ satisfies $\chi(G) \leq \gall(G)$.
\end{conjecture}

Our first piece of evidence in support of this conjecture is the fact that the fractional relaxation holds:

\begin{theorem}[Superlocal fractional relaxation]\label{thm:frac3}
Every graph $G$ satisfies $\chi_f(G) \leq \gall'(G)$.
\end{theorem}

After proving this theorem, we will prove that Conjecture \ref{con:3} holds for graphs with no stable set of size 3.  We then prove that Conjecture \ref{con:3} holds for line graphs and quasi-line graphs.  The proofs closely follow the proofs of the Local Reed's Conjecture for the corresponding graph classes, which appear in \cite{chudnovskykps12} and \cite{kingthesis}.

Before proving Theorem \ref{thm:frac3}, we describe our original motivation.  In \cite{edwardsk13} we bound the fractional chromatic number of $K_\Delta$-free graphs.  Our approach is to find a partial fractional colouring of one type, then use an extension of Theorem \ref{thm:frac2} as a ``finishing blow'' to complete the colouring.  So the question naturally arises: can we strengthen the finishing blow?  Although Theorem \ref{thm:frac3} does not improve the results given in \cite{edwardsk13}, we feel that Conjecture \ref{con:3} is of greater interest. In Section \ref{sec:conc} we discuss possible extensions of Theorem \ref{thm:frac2} that would in fact strengthen these previous results.

\section{Proving the fractional relaxation}\label{sec:frac}

The proofs of Theorems \ref{thm:frac1}, \ref{thm:frac2}, and \ref{thm:frac3} all rely on the same natural fractional colouring algorithm, originally due to Reed \cite{molloyrbook}: we add equal weight to every maximum stable set until a vertex is completely coloured, then we discard all completely coloured vertices and continue the process, respecting the fact that discarding vertices changes the set of maximum stable sets.  Improving the bounds we get is merely a matter of refining the analysis.  Before describing this process in greater detail we give some requisite definitions.

For a graph $G$ and a nonnegative rational $k$, a {\em fractional $k$-colouring} of $G$ is a nonnegative weighting $w$ on the stable sets of $G$ such that $\sum_{S}w(S) \leq k$, and for every vertex $v$, $\sum_{S \ni v} w(S) = 1$.  The {\em fractional chromatic number} of $G$, written $\chi_f(G)$, is the smallest $k$ for which $G$ has a fractional vertex $k$-colouring.

The proof of Theorem \ref{thm:frac2} relies on the following lemma, whose proof appears in \S2.2 of \cite{kingthesis}.

\begin{lemma}\label{lem:exp1}
Let $S$ be a maximum stable set of $G$ chosen uniformly at random.  Then for any vertex $v$, $\E(|S\cap N(v)|) \geq 2-(\omega(v) + 1)\p(v\in S)$.
\end{lemma}

Before proving Theorem \ref{thm:frac3} we need an easy generalization.  For adjacent vertices $u$ and $v$ we define $N(u,v)$ as $(N(u)\cup N(v))\setminus\{u,v\}$.

\begin{lemma}\label{lem:expectation}
Let $S$ be a maximum stable set of $G$ chosen uniformly at random.  Then for any adjacent vertices $u$ and $v$,
\begin{equation}\label{eq:exp}
\E(|S\cap N(u,v)|) \geq 4- (\omega(v) + 2)\p(v\in S)-(\omega(u)+2)\p(u\in S)  - \sum_{w\in N(v) \cap N(u)}\p(w\in S).
\end{equation}
\end{lemma}
\begin{proof}
We know by Lemma \ref{lem:exp1} that 
\begin{eqnarray}
\label{eq:nbrv}
\E(|S\cap N(v)|) &\geq& 2-(\omega(v) + 1)\p(v\in S)\textrm{ and}\\
\label{eq:nbru}
\E(|S\cap N(u)|) &\geq& 2-(\omega(u) + 1)\p(u\in S).
\end{eqnarray}
By linearity of expectation we have
\begin{equation}
\label{eq:lin}
\E(|S\cap N(u,v)|) = \E(|S\cap N(u)|) + \E(|S\cap N(v)|) - \E(|S\cap \tilde N(u)\cap \tilde N(v)|).
\end{equation}
Also by  linearity of expectation, we have
\begin{equation}
\label{eq:nbr}
\E(|S\cap \tilde N(u)\cap \tilde N(v)|) = \p(u\in S) + \p(v\in S ) + \sum_{w\in N(v) \cap N(u)}\p(w\in S).
\end{equation} 
Substituting (\ref{eq:nbrv}), (\ref{eq:nbru}), and (\ref{eq:nbr}) into (\ref{eq:lin}) gives us (\ref{eq:exp}).
\end{proof}

We are now ready to prove Theorem \ref{thm:frac3}.

\begin{proof}[Proof of Theorem \ref{thm:frac3}]We fractionally colour $G$ using the following iterative method.

\begin{enumerate}
\item Set $w(S) = 0$ for every $S \in \fs$.  Set $G_0 = G$.  Set $i=0$.

Set $T=0$.  $T$ stands for total weight used.

For each $v \in V$, set $\textit{wo}_v = 0$ ($\textit{wo}$ stands for weight on).

\item If $V(G_i) = \emptyset$ or $T = \gall'(G)$ then stop.

\item For each vertex $v$ of $G_i$, let $p_i(v)$ be the probability that $v$ is in a uniformly random maximum stable set of $G_i$.  Set $\textit{low} = \min\{\frac{1-\textit{wo}_v}{p_i(v)}|v \in V(G_i) \}$.  Set $\textit{val}_i = \min(\textit{low}, \gall'(G) - T)$. 

\item Let $\fs_i$ be the set of maximum stable sets of $G_i$.  For each stable set in $\fs_i$, increase $w(S)$ by $\frac{\textit{val}_i}{|\fs_i|}$.  For each vertex $v$ of $G_i$, increase $\textit{wo}_v$ by $p_i(v)\textit{val}_i$.  Increase $T$ by $\textit{val}_i$.

\item Let $G_{i+1}$ be the graph induced by those vertices $v$ which satisfy $\textit{wo}_v<1$.  Increment $i$ and go to Step 2.

\end{enumerate}

Our choice of $\textit{val}_i$ ensures two things:  that $T$ never exceeds $\gall'(G)$, and that if the $i$th iteration is not the last, then $V(G_{i+1})$ is properly contained in $V(G_i)$.  Thus the algorithm must terminate.

We claim that at the end of the procedure, the $w(S)$ weights give a fractional $\gall'(G)$-colouring.  It is easy to show by induction that at the end of each iteration and for every $v \in V$, $\textit{wo}_v = \sum_{\{S \in \fs | v \in S\}}w(S)$ and $T = \sum_{S\in \fs}w(S)$.  The definitions of $\textit{low}$ and $\textit{val}_i$ ensure that no $\textit{wo}_v$ is ever more than 1.  We stop if $V(G_i) = \emptyset$ or $T = \gall'(G)$; in the first case we know that we have the desired fractional colouring.  We must now show that the same is true in the second case.  It suffices to show that in this case, each $\textit{wo}_v = 1$.

So assume that for some $v$ we have $\textit{wo}_v <1$ when we complete the process.  For each vertex $u$ and iteration $i$, denote by $a_i(u)$ the amount by which $\textit{wo}_u$ was augmented in iteration $i$, i.e.\ $a_i(u) = \textit{val}_ip_i(u)$. There are two cases; we will show that each results in a contradiction.

\vspace{.5em}\noindent{\bf Case 1}: $v$ has a neighbour $u$ with $wo_u < 1$.

In this case $\{u,v\} \subseteq V(G_i)$ for every $i$. For every $i$, let $S$ be a maximum stable set drawn at random from $\mathcal S_i$. Then by Lemma \ref{lem:expectation},
$$ val_i\E(|S\cap N(u,v)|) = \sum_{x\in N(u,v)}a_i(x) \geq 4val_i - (\omega(v) + 2)a_i(v) - (\omega(u) + 2)a_i(u) - \sum_{w\in N(u)\cap N(v)}a_i(w) $$
Summing over all iterations,

	\begin{align*}
		\sum_{x\in N(u,v)}wo_x &\geq 4T - (\omega(v) + 2)wo_v - (\omega(u) + 2)wo_u - \sum_{w\in N(u)\cap N(v)}wo_w \\
		&> \omega(u) + \omega(v) + d(u) + d(v) + 2 - (\omega(v) + 2) - (\omega(u) + 2) - |N(u)\cap N(v)| \\
		&= d(u) + d(v) - |N(u)\cap N(v)| - 2 = |N(u,v)|,
	\end{align*}
a contradiction since $wo_x \leq 1$ for each $x \in N(u,v)$.

\vspace{.5em}\noindent{\bf Case 2}: Every neighbour $u$ of $v$ has $wo_u = 1$ at the end of the procedure.

For every neighbour $u$ of $v$ there exists some $j$ such that $u\in V(G_j)$ but $u\notin V(G_{j+1})$.  Choose $u$ maximizing $j$; this implies that $N_{G_i}(v) = \emptyset$ for all $i>j$, and consequently $a_i(v) = val_i$ for each $i>j$.
When $i\leq j$ we again have  $$ \sum_{x\in N(u,v)}a_i(x) \geq 4val_i - (\omega(v) + 2)a_i(v) - (\omega(u) + 2)a_i(u) - \sum_{w\in N(u)\cap N(v)}a_i(w) $$ by Lemma \ref{lem:expectation}.
Summing over the iterations up to $j$ we see
\footnotesize
\begin{eqnarray*}
&&\sum_{x\in N(u,v)} \sum_{i\leq j} a_i(x)  \\
&\geq& 4(T- \sum_{i>j}a_i(v)) - (\omega(v) + 2)\sum_{i\leq j}a_i(v) - (\omega(u) + 2)\sum_{i\leq j}a_i(u) - \sum_{w\in N(u)\cap N(v)}\sum_{i\leq j}a_i(w) \\
&=& d(u) + d(v) + \omega(u) + \omega(v) + 2 - 4\sum_{i>j}a_i(v) - (\omega(v) + 2)\sum_{i\leq j}a_i(v) - (\omega(u) + 2)\sum_{i\leq j}a_i(u) - \sum_{w\in N(u)\cap N(v)}\sum_{i\leq j}a_i(w) \\
&\geq& d(u) + d(v) + \omega(u) + \omega(v) + 2 - (\omega(v) + 2)\sum_{i}a_i(v) - (\omega(u) + 2)wo_u - \sum_{w\in N(u)\cap N(v)}wo_w \\
&>& d(u) + d(v) - |N(u)\cap N(v)| - 2 = |N(u,v)|,
\end{eqnarray*}
\normalsize
where the third inequality follows since $\omega(v) +2\geq 4$. This is a contradiction as $wo_x \leq 1$ for each $x \in N(u,v)$.

It follows that for every $v \in V(G)$, $\textit{wo}_v =1$.  This completes the proof.
\end{proof}

\section{Some easy integer colouring cases}

Theorem \ref{thm:frac3} puts Conjecture \ref{con:3} within reach for several classes of graphs.  For circular interval graphs (see \cite{kingthesis} for a definition), the result is an immediate consequence of the {\em round-up property} proved by Niessen and Kind \cite{niessenk00}:

\begin{theorem}
For any circular interval graph $G$, $\chi(G) = \lceil \chi_f(G)\rceil$.
\end{theorem}

\begin{theorem}\label{thm:cig}
For any circular interval graph $G$, $\chi(G)\leq \gall(G)$.
\end{theorem}

Circular interval graphs are a fundamental subclass of quasi-line graphs, which are themselves a fundamental subclass of claw-free graphs -- see \cite{cssurvey} for an explanation.  Since Reed's Conjecture is known to hold for claw-free graphs, we might hope that the same is true for Conjecture \ref{con:3}.  Although there are still some claw-free graphs for which Conjecture \ref{con:2} has not been proven, we hope to prove the superlocal Reed's Conjecture for substantial subclasses of claw-free graphs.  We continue by naming another easy victim: graphs with stability number at most two.

\begin{theorem}
Any graph $G$ satisfying $\alpha(G)\leq 2$ also satisfies $\chi(G)\leq \gall(G)$.
\end{theorem}

In this case a colouring of $G$ corresponds to a matching in the complement of $G$, so we have a wealth of knowledge at hand.  The proof of this theorem is actually an easy exercise, and follows almost exactly the proof of Theorem 2.15 in \cite{kingthesis}.  The Edmonds-Gallai structure theorem \cite{edmonds65b, gallai59} implies that a minimum counterexample, which must be vertex-critical, either satisfies $\chi(G) = \lceil\chi_f(G)\rceil$ or has a disconnected complement.  Thus the only work we need to do, after replacing Theorem \ref{thm:frac2} with Theorem \ref{thm:frac3}, is to prove that if $G$ is the join of graphs $G_1$ and $G_2$, then $\gall(G) \geq \gall(G_1)+\gall(G_2)$.  We leave the details to the reader.

Having exhibited the usefulness of Theorem \ref{thm:frac3} in bounding the chromatic number, we move on to something a little more challenging: a class of graphs for which $\chi$ and $\chi_f$ are believed, but not known, to differ by at most 1.

\section{Colouring line graphs with $\gall(G)$ colours}

In this section we consider line graphs of multigraphs.  As we do, we bear in mind the famous Goldberg-Seymour conjecture \cite{goldberg73,seymour79}, which proposes that every line graph $G$ satisfies $\chi(G) \leq \chi_f(G)+1$.  Kahn \cite{kahn96} proved that this bound holds asymptotically.  The approach used to prove Conjecture \ref{con:1} for line graphs \cite{kingrv07} was no help in proving Conjecture \ref{con:2}.  We therefore appeal to Vizing fans, which were the key to proving Conjecture \ref{con:2} for line graphs \cite{chudnovskykps12}.  The extension of this proof is fairly straightforward.

In order to prove Conjecture \ref{con:3} for line graphs, we prove an equivalent statement in the setting of edge colourings of multigraphs.  Given distinct adjacent vertices $u$ and $v$ in a multigraph $G$, we let $\mu_G(uv)$ denote the number of edges between $u$ and $v$.  We let $t_G(uv)$ denote the maximum, over all vertices $w$ forming a triangle with $ \{u,v\}$, of the number of edges with both endpoints in $\{u,v,w\}$.  That is,
$$t_G(uv) := \max_{w\in N(u)\cap N(v)}\left(\mu_G(uv)+\mu_G(uw)+\mu_G(vw)\right).$$
We omit the subscripts when the multigraph in question is clear.

Observe that given an edge $e$ in $G$ with endpoints $u$ and $v$, the degree of $e$ in $L(G)$ is $d(u)+d(v)-\mu(uv)-1$.  And since any clique in $L(G)$ containing $e$ comes from the edges incident to $u$, the edges incident to $v$, or the edges in a triangle containing $u$ and $v$, we can see that $\omega(e)$ in $L(G)$ is equal to $\max \{ d(u), d(v), t(uv) \}$.  Therefore we prove the following theorem, which is equivalent to proving Conjecture \ref{con:3} for line graphs:

\begin{theorem}\label{thm:line}
Let $G$ be a multigraph and let
\small
\begin{eqnarray}
\bar \gall(G) := \bigg\lceil \tfrac 12\max_{uv, vw\in E(G)} \big\{ &&d(u) + \tfrac 12(d(v) - \mu(uv)) + d(v) + \tfrac 12 (d(w)-\mu(vw)), \\ 
												&&d(u) + \tfrac 12(d(v) - \mu(uv)) + d(w) + \tfrac 12 (d(v)-\mu(vw)), \label{eqn:2}\\ 
												&&d(u) + \tfrac 12(d(v) - \mu(uv)) + \tfrac 12(d(v)+d(w)-\mu(vw)+t(vw)), \\ 
												&&d(v) + \tfrac 12(d(u) - \mu(uv)) + d(v) + \tfrac 12 (d(w)-\mu(vw)), \label{eqn:4} \\
												&&d(v) + \tfrac 12(d(u) - \mu(uv)) + d(w) + \tfrac 12 (d(v)-\mu(vw)), \\
												&&d(v) + \tfrac 12(d(u) - \mu(uv)) + \tfrac 12(d(v)+d(w)-\mu(vw)+t(vw)), \\
												&&\tfrac 12 (d(u)+d(v)-\mu(uv)+t(uv)) + d(v) + \tfrac 12 (d(w)-\mu(vw)),\\
												&&\tfrac 12 (d(u)+d(v)-\mu(uv)+t(uv)) + d(w) + \tfrac 12 (d(v)-\mu(vw)),\\
												&&\tfrac 12 (d(u)+2d(v)+d(w)-\mu(uv)+t(uv) -\mu(vw)+t(vw)) \label{eqn:9}\\ 
												\big\} \bigg\rceil. \nonumber
\end{eqnarray}

 Then $\chi'(G) \leq \bar \gall(G)$. 
\end{theorem}

\noindent{\bf Remark:} One can turn the proof of this theorem into an algorithm as in \cite{chudnovskykps12}, yielding an $O(n^2)$ algorithm for $\gall(G)$-colouring a line graph on $n$ vertices.  In fact, what we implicitly prove is that the algorithm presented in $\cite{chudnovskykps12}$ gives a $\gall(G)$-colouring, not just a $\gamma_\ell(G)$-colouring.\\

To prove this theorem we assume that $G$ is a minimum counterexample and investigate $\bar \gall(G)$-edge-colourings of $G-e$ for an edge $e$.  We begin by defining, for a vertex $v$, a {\em fan hinged at $v$}.  Let $e$ be an edge incident to $v$, and let $v_1,\ldots, v_\ell$ be a set of distinct neighbours of $v$ with $e$ between $v$ and $v_1$.  Let $c:E\setminus \{e\} \rightarrow \{1,\ldots,k\}$ be a proper edge colouring of $G\setminus \{e\}$ for some fixed $k$.  Then $F = (e;c;v;v_1,\ldots, v_\ell)$ is a {\em fan} if for every $j$ such that $2\leq j \leq \ell$, there exists some $i$ less than $j$ such that some edge between $v$ and $v_j$ is assigned a colour that does not appear on any edge incident to $v_i$ (i.e.\ a colour {\em missing} at $v_i$).  We say that $F$ is {\em hinged at $v$}.  If there is no $u \notin \{ v,v_1,\ldots,v_\ell \}$ such that $F'=(e;c;v;v_1,\ldots,v_\ell,u)$ is a fan, we say that $F$ is a {\em maximal fan}.  The {\em size} of a fan refers to the number of neighbours of the hinge vertex contained in the fan (in this case, $\ell$).  These fans generalize Vizing's fans, originally used in the proof of Vizing's theorem \cite{vizing64}.  Given a partial $k$-edge-colouring of $G$ and a vertex $w$, we say that a colour is {\em incident to $w$} if the colour appears on an edge incident to $w$.  We use $\fc(w)$ to denote the set of colours incident to $w$, and we use $\bar\fc(w)$ to denote $[k] \setminus \fc(w)$.

For this section let us call $G$ a minimum counterexample if $\chi'(G)> \bar \gall(G)$ and for every graph $G'$ on fewer edges, $\chi'(G')\leq \bar \gall(G')$. Fans allow us to modify partial $k$-edge-colourings of a graph (specifically those with exactly one uncoloured edge). As a first step towards Theorem \ref{thm:line}, we show that if $G$ is a minimum counterexample and $k = \bar \gall(G)$, then every maximal fan has size $2$. For ease of notation we will denote $\bar \gall(G)$ by $k$ for the remainder of this section. We begin with two simple lemmas that guarantee disjointness of certain colour sets in partial $k$-edge-colourings of $G-e$. These follow from the work of Vizing \cite{vizing64}; for proofs see for example Lemmas 6 and 7 in \cite{chudnovskykps12}.

\begin{lemma}\label{lem:algo1}
Let $G$ be a minimum counterexample, let $e$ be an edge in $G$ and let $c$ be a $k$-edge-colouring of $G-e$. If $F = (e;c;v;v_1,\ldots,v_\ell)$ is a fan, then $\bar \fc(v)\cap \bar \fc(v_j) = \emptyset$ for every $j$.
\end{lemma}

\begin{lemma}\label{lem:algo2}
Let $G$ be a minimum counterexample, let $e$ be an edge in $G$ and let $c$ be a $k$-edge-colouring of $G-e$. If $F = (e;c;v;v_1,\ldots,v_\ell)$ is a fan, then for every $i$ and $j$ satisfying $1\leq i<j\leq \ell$, $\bar \fc(v_i)\cap \bar \fc(v_j) = \emptyset$.
\end{lemma}

We can now prove that no maximal fan has size $1$ or at least $3$.

\begin{lemma}\label{lem:single}
Let $G$ be a minimum counterexample, let $e$ be an edge in $G$ and let $c$ be a $k$-edge-colouring of $G-e$. Let $F=(e;c;v;v_1,v_2,\ldots,v_\ell)$ be a maximal fan. Then $\ell >1$.
\end{lemma}

\begin{proof}
Suppose that $\ell =1$. If $\bar\fc(v)\cap \bar\fc(v_1)$ is nonempty, then $c$ can easily be extended to a $k$-edge-colouring of $G$, so we may assume $\bar\fc(v) \cap \bar\fc(v_1)$ is empty.   Now, $k = \bar \gall(G) \geq d(v_1)$ by (\ref{eqn:4}) and so $\bar{\mathcal{C}}(v_1)$ is nonempty.  Therefore there is a colour in $\bar\fc(v_1)$ appearing on an edge incident to $v$ whose other endpoint, call it $v_2$, is not $v_1$.  Thus $(e;c;v;v_1,v_2)$ is a fan, contradicting the maximality of $F$.
\end{proof}

\begin{lemma}\label{lem:algo3}
Let $G$ be a minimum counterexample, let $e$ be an edge in $G$ and let $c$ be a $k$-edge-colouring of $G-e$. Let $F=(e;c;v;v_1,v_2,\ldots,v_\ell)$ be a maximal fan. Then $\ell <3$.
\end{lemma}

\begin{proof}
Suppose $\ell \geq 3$. Let $v_0$ denote $v$ for ease of notation.  If the sets $\bar\fc(v_0), \bar\fc(v_1), \ldots, \bar\fc(v_\ell)$ are not all pairwise disjoint, then using Lemma \ref{lem:algo1} or Lemma \ref{lem:algo2} we can find a $k$-edge-colouring of $G$, contradicting $\chi'(G)>k$. We therefore assume they are all pairwise disjoint.

The number of missing colours at $v_i$, i.e.\ $|\bar \fc(v_i)|$, is $k-d(v_i)$ if $2\leq i\leq \ell$, and $k-d(v_i)+1$ if $i\in \{0,1\}$.  Since $F$ is maximal, any edge with one endpoint $v_0$ and the other endpoint outside $\{v_0,\ldots,v_\ell\}$ must have a colour not appearing in $\cup_{i=0}^\ell\bar\fc(v_i)$. Therefore 
\begin{equation}
\left( \sum_{i=0}^{\ell} \left( k-d(v_i) \right)\right)+2 + \left( d(v_0) - \sum_{i=1}^{\ell}\mu(v_0v_i) \right) \leq k.
\end{equation}

\begin{equation}
\ell k + 2 - \sum_{i=1}^{\ell}\mu(v_0v_i) \leq \sum_{i=1}^{\ell}d(v_i).
\end{equation}

\noindent But since $k = \bar \gall(G)$ we have for each $i$, 
\begin{equation}
2k \geq d(v_i) + d(v_{i+1}) + d(v_0) - \tfrac 12 (\mu(v_0v_i) + \mu(v_0v_{i+1}))
\end{equation}
by (\ref{eqn:2}), taking indices modulo $\ell$. This tells us that

\begin{eqnarray}
2\ell k &\geq&  \sum_{i=1}^{\ell} \big( {d(v_i) + d(v_{i+1}) + d(v_0) - \tfrac 12(\mu(v_0v_i) + \mu(v_0v_{i+1}))}\big)\\
&=& \sum_{i=1}^{\ell} \big( 2d(v_i) + d(v_0) - \mu(v_0v_i) \big)
\end{eqnarray}

\noindent so we have 

\begin{equation}
\tfrac 12 \sum_{i=1}^{\ell} \big({2d(v_i) + d(v_0) - \mu(v_0v_i)}\big) + 2 - \sum_{i=1}^{\ell}\mu(v_0v_i) \leq \sum_{i=1}^{\ell}d(v_i).
\end{equation}

\noindent But then
\begin{eqnarray*}
2 + \tfrac 12 \ell d(v_0) - \tfrac 32 \sum_{i=1}^{\ell}\mu(v_0v_i) &\leq& 0 \\
\tfrac{\ell}{2}d(v_0) &<& \tfrac 32 d(v_0),
\end{eqnarray*}

\noindent a contradiction since $\ell \geq 3$.
\end{proof}

We are now ready to finish the proof of Theorem \ref{thm:line}. We approach the theorem by constructing a sequence of overlapping fans of size two until we can apply a previous lemma.  If we cannot do this, then our sequence results in a cycle in $G$ and a set of partial $k$-edge-colourings of $G$ with a very specific structure that leads us to a contradiction.

\begin{proof}[Proof of Theorem \ref{thm:line}]
Let $G$ be a minimum counterexample and let $e_0$ be an edge of $G$. Let $c_0$ be a $k$-edge-colouring of $G-e$.

Let $v_0$ and $v_1$ be the endpoints of $e_0$, and let $F_0=(e_0;c_0;v_1;v_0,v_2)$ be a maximal fan whose existence and maximality are guaranteed by Lemmas \ref{lem:single} and \ref{lem:algo3}. 

Let $\bar\fc_0$ denote the set of colours missing at $v_0$ in the partial colouring $c_0$, and take some colour $\alpha_0 \in \bar\fc_0$.  Note that if $\alpha_0$ does not appear on an edge between $v_1$ and $v_2$ then we can find a fan $(e_0;c_0;v_1;v_0,v_2,u)$ of size $3$, contradicting Lemma \ref{lem:algo3}.  So we can assume that $\alpha_0$ does appear on an edge between $v_1$ and $v_2$.

Let $e_1$ denote the edge between $v_1$ and $v_2$ given colour $\alpha_0$ in $c_0$.  We construct a new colouring $c_1$ of $G-e_1$ from $c_0$ by uncolouring $e_1$ and assigning $e_0$ colour $\alpha_0$.  Let $\bar \fc_1$ denote the set of colours missing at $v_1$ in the colouring $c_1$.  Now  let $F_1=(e_1;c_1;v_2;v_1,v_3)$ be a maximal fan.  As with $F_0$, we can assume that $F_1$ exists and is indeed maximal.  The vertex $v_3$ may or may not be the same as $v_0$.

Let $\alpha_1\in\bar\fc_1$ be a colour in $\bar\fc_1$.  Just as $\alpha_0$ appears between $v_1$ and $v_2$ in $c_0$, we can see that $\alpha_1$ appears between $v_2$ and $v_3$.  Now let $e_2$ be the edge between $v_2$ and $v_3$ having colour $\alpha_1$ in $c_1$.  We construct a colouring $c_2$ of $G-e_2$ from $c_1$ by uncolouring $e_2$ and assigning $e_1$ colour $\alpha_1$.  

\begin{figure}[ht]	
	\hspace{-1.2cm}{
	\begin{picture}(0,0)%
\includegraphics{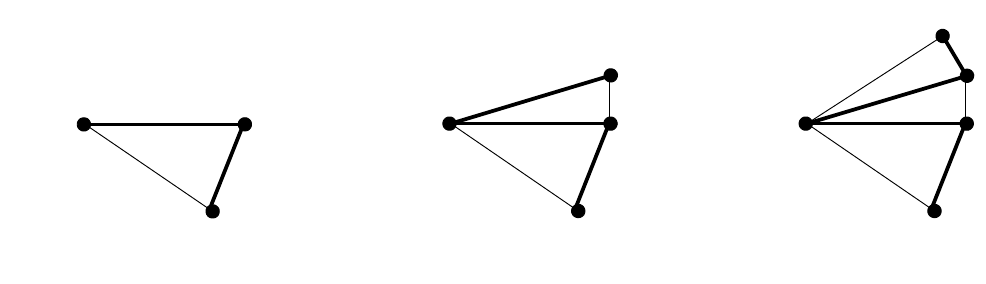}%
\end{picture}%
\setlength{\unitlength}{3947sp}%
\begingroup\makeatletter\ifx\SetFigFont\undefined%
\gdef\SetFigFont#1#2#3#4#5{%
  \reset@font\fontsize{#1}{#2pt}%
  \fontfamily{#3}\fontseries{#4}\fontshape{#5}%
  \selectfont}%
\fi\endgroup%
\begin{picture}(7872,2259)(2986,-1573)
\put(8776,-163){\makebox(0,0)[lb]{\smash{{\SetFigFont{12}{14.4}{\familydefault}{\mddefault}{\updefault}{\color[rgb]{0,0,0}$\alpha_0 \in \bar{\mathcal C_2}$}%
}}}}
\put(4468,-236){\makebox(0,0)[lb]{\smash{{\SetFigFont{12}{14.4}{\familydefault}{\mddefault}{\updefault}{\color[rgb]{0,0,0}$e_0$}%
}}}}
\put(10243,-229){\makebox(0,0)[lb]{\smash{{\SetFigFont{12}{14.4}{\familydefault}{\mddefault}{\updefault}{\color[rgb]{0,0,0}$e_2$}%
}}}}
\put(5926,-160){\makebox(0,0)[lb]{\smash{{\SetFigFont{12}{14.4}{\familydefault}{\mddefault}{\updefault}{\color[rgb]{0,0,0}$\alpha_1 \in \bar{\mathcal C_1}$}%
}}}}
\put(6994,-1495){\makebox(0,0)[lb]{\smash{{\SetFigFont{12}{14.4}{\familydefault}{\mddefault}{\updefault}{\color[rgb]{0,0,0}$F_1$}%
}}}}
\put(7927,-427){\makebox(0,0)[lb]{\smash{{\SetFigFont{12}{14.4}{\familydefault}{\mddefault}{\updefault}{\color[rgb]{0,0,0}$v_2$}%
}}}}
\put(7393,-1162){\makebox(0,0)[lb]{\smash{{\SetFigFont{12}{14.4}{\familydefault}{\mddefault}{\updefault}{\color[rgb]{0,0,0}$v_3$}%
}}}}
\put(7792,-760){\makebox(0,0)[lb]{\smash{{\SetFigFont{12}{14.4}{\familydefault}{\mddefault}{\updefault}{\color[rgb]{0,0,0}$e_2$}%
}}}}
\put(6394,-427){\makebox(0,0)[lb]{\smash{{\SetFigFont{12}{14.4}{\familydefault}{\mddefault}{\updefault}{\color[rgb]{0,0,0}$v_1$}%
}}}}
\put(7393,-229){\makebox(0,0)[lb]{\smash{{\SetFigFont{12}{14.4}{\familydefault}{\mddefault}{\updefault}{\color[rgb]{0,0,0}$e_1$}%
}}}}
\put(7258, 38){\makebox(0,0)[lb]{\smash{{\SetFigFont{12}{14.4}{\familydefault}{\mddefault}{\updefault}{\color[rgb]{0,0,0}$e_0$}%
}}}}
\put(7993, 38){\makebox(0,0)[lb]{\smash{{\SetFigFont{12}{14.4}{\familydefault}{\mddefault}{\updefault}{\color[rgb]{0,0,0}$v_0$}%
}}}}
\put(7993,-1027){\makebox(0,0)[lb]{\smash{{\SetFigFont{12}{14.4}{\familydefault}{\mddefault}{\updefault}{\color[rgb]{0,0,0}$c_1(e_2) = \alpha_1$}%
}}}}
\put(3001,-170){\makebox(0,0)[lb]{\smash{{\SetFigFont{12}{14.4}{\familydefault}{\mddefault}{\updefault}{\color[rgb]{0,0,0}$\alpha_0 \in \bar{\mathcal C_0}$}%
}}}}
\put(4468,-1164){\makebox(0,0)[lb]{\smash{{\SetFigFont{12}{14.4}{\familydefault}{\mddefault}{\updefault}{\color[rgb]{0,0,0}$v_2$}%
}}}}
\put(5002,-433){\makebox(0,0)[lb]{\smash{{\SetFigFont{12}{14.4}{\familydefault}{\mddefault}{\updefault}{\color[rgb]{0,0,0}$v_1$}%
}}}}
\put(4867,-766){\makebox(0,0)[lb]{\smash{{\SetFigFont{12}{14.4}{\familydefault}{\mddefault}{\updefault}{\color[rgb]{0,0,0}$e_1$}%
}}}}
\put(4069,-1495){\makebox(0,0)[lb]{\smash{{\SetFigFont{12}{14.4}{\familydefault}{\mddefault}{\updefault}{\color[rgb]{0,0,0}$F_0$}%
}}}}
\put(3469,-433){\makebox(0,0)[lb]{\smash{{\SetFigFont{12}{14.4}{\familydefault}{\mddefault}{\updefault}{\color[rgb]{0,0,0}$v_0$}%
}}}}
\put(5068,-1030){\makebox(0,0)[lb]{\smash{{\SetFigFont{12}{14.4}{\familydefault}{\mddefault}{\updefault}{\color[rgb]{0,0,0}$c_0(e_1) = \alpha_0$}%
}}}}
\put(10108, 38){\makebox(0,0)[lb]{\smash{{\SetFigFont{12}{14.4}{\familydefault}{\mddefault}{\updefault}{\color[rgb]{0,0,0}$e_1$}%
}}}}
\put(9844,-1495){\makebox(0,0)[lb]{\smash{{\SetFigFont{12}{14.4}{\familydefault}{\mddefault}{\updefault}{\color[rgb]{0,0,0}$F_2$}%
}}}}
\put(9244,-429){\makebox(0,0)[lb]{\smash{{\SetFigFont{12}{14.4}{\familydefault}{\mddefault}{\updefault}{\color[rgb]{0,0,0}$v_2$}%
}}}}
\put(10243,-1162){\makebox(0,0)[lb]{\smash{{\SetFigFont{12}{14.4}{\familydefault}{\mddefault}{\updefault}{\color[rgb]{0,0,0}$v_4$}%
}}}}
\put(10777,-429){\makebox(0,0)[lb]{\smash{{\SetFigFont{12}{14.4}{\familydefault}{\mddefault}{\updefault}{\color[rgb]{0,0,0}$v_3$}%
}}}}
\put(10642,-762){\makebox(0,0)[lb]{\smash{{\SetFigFont{12}{14.4}{\familydefault}{\mddefault}{\updefault}{\color[rgb]{0,0,0}$e_3$}%
}}}}
\put(10642,239){\makebox(0,0)[lb]{\smash{{\SetFigFont{12}{14.4}{\familydefault}{\mddefault}{\updefault}{\color[rgb]{0,0,0}$e_0$}%
}}}}
\put(10510,503){\makebox(0,0)[lb]{\smash{{\SetFigFont{12}{14.4}{\familydefault}{\mddefault}{\updefault}{\color[rgb]{0,0,0}$v_0$}%
}}}}
\put(10843, 38){\makebox(0,0)[lb]{\smash{{\SetFigFont{12}{14.4}{\familydefault}{\mddefault}{\updefault}{\color[rgb]{0,0,0}$v_1$}%
}}}}
\put(10843,-1027){\makebox(0,0)[lb]{\smash{{\SetFigFont{12}{14.4}{\familydefault}{\mddefault}{\updefault}{\color[rgb]{0,0,0}$c_2(e_3) = \alpha_0$}%
}}}}
\end{picture}%
}%
	\caption{Construction of the first few fans $F_i$}
	\label{fig:fans}
\end{figure}
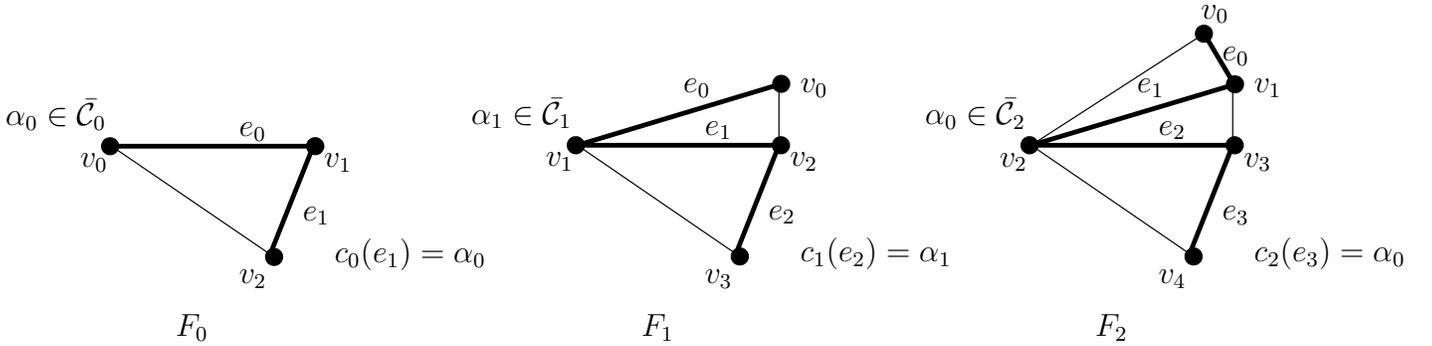

We continue to construct a sequence of fans $F_i = (e_i,c_i;v_{i+1};v_i,v_{i+2})$ for $i=0,1,2,\ldots$ in this way, maintaining the property that $\alpha_{i+2}=\alpha_i$ (see Figure \ref{fig:fans}).
This is possible because when we construct $c_{i+1}$ from $c_i$, we make $\alpha_i$ available at $v_{i+2}$, so the set $\bar\fc_{i+2}$ (the set of colours missing at $v_{i+2}$ in the colouring $c_{i+2}$) always contains $\alpha_{i}$.  We continue constructing our sequence of fans until we reach some $j$ for which $v_j \in \{v_i\}_{i=0}^{j-1}$, which will inevitably happen if we never find a fan of size 3 or greater.  We claim that $v_j=v_0$ and $j$ is odd.  To see this, consider the original edge-colouring of $G-e_0$ and note that for $1\leq i\leq j-1$, $\alpha_{0}$ appears on an edge between $v_i$ and $v_{i+1}$ precisely if $i$ is odd, and $\alpha_1$ appears on an edge between $v_i$ and $v_{i+1}$ precisely if $i$ is even.  Thus since the edges of colour $\alpha_0$ form a matching, and so do the edges of colour $\alpha_1$, we indeed have $v_j=v_0$ and $j$ odd.  Furthermore $F_0=F_j$.  Let $C$ denote the cycle $v_0,v_1,\ldots,v_{j-1}$.  In each colouring, $\alpha_0$ and $\alpha_1$ both appear $(j-1)/2$ times on $C$, in a near-perfect matching.  Let $H$ be the sub-multigraph of $G$ consisting of those edges between $v_i$ and $v_{i+1}$ for $0\leq i\leq j$ (with indices modulo $j$).  Let $A$ be the set of colours missing on at least one vertex of $C$, and let $H_A$ be the sub-multigraph of $H$ consisting of $e_0$ and those edges receiving a colour in $A$ in $c_0$ (and therefore in any $c_i$).

Suppose $j=3$. If some colour is missing on two vertices of $C$ in $c_0$, $c_1$, or $c_2$, we
can easily find a $k$-edge-colouring of $G$ since any two vertices of $C$ are the endpoints
of $e_0$, $e_1$, or $e_2$, a contradiction since $G$ is a minimum counterexample. We know that every colour in $\bar{\mathcal C_0}$ appears between $v_1$ and $v_2$, and
every colour in $\bar{\mathcal C_1}$ appears between $v_2$ and $v_0$ and every colour in $\bar{\mathcal C_2}$ appears between $v_0$ and $v_1$. Therefore $|E(H_A)| = |A| + 1$ and by (\ref{eqn:9}) we have
\begin{eqnarray*}
&&4\bar \gall(G)\\
&\geq& d_G(v_1) + d_G(v_2) + 2d_G(v_0) - \mu_G(v_0v_1) - \mu_G(v_0v_2) + t_G(v_0v_1) + t_G(v_0v_2) \\
			&=&  d_{H_A}(v_1) + d_{H_A}(v_2) + 2d_{H_A}(v_0) + 4(k-|A|) - \mu_G(v_0v_1) - \mu_G(v_0v_2) + t_G(v_0v_1) + t_G(v_0v_2) \\
			&\geq& d_{H_A}(v_1) + d_{H_A}(v_2) + 2d_{H_A}(v_0) + 4(k-|A|) - \mu_{H_A}(v_0v_1) - \mu_{H_A}(v_0v_2) + t_{H_A}(v_0v_1) + t_{H_A}(v_0v_2) \\
			&\geq& 4|E(H_A)| + 4(k-|A|) \\
			&>& 4|A| + 4(k-|A|) = 4k,
\end{eqnarray*}
a contradiction since $k = \bar \gall(G)$.  We can therefore assume that $j\geq 5$.

Let $\beta$ be a colour in $A \setminus \{\alpha_0,\alpha_1\}$.  If $\beta$ is missing at two consecutive vertices $v_i$ and $v_{i+1}$ then we can easily extend $c_i$ to a $k$-edge-colouring of $G$.  Bearing in mind that each $F_i$ is a maximal fan, we claim that if $\beta$ is not missing at two consecutive vertices then either we can easily $k$-edge-colour $G$, or the number of edges coloured $\beta$ in $H_{A}$ is at least twice the number of vertices at which $\beta$ is missing in any $c_i$.

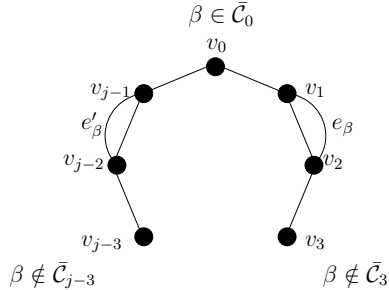
\begin{figure}[ht]
	\centering
	\scalebox{0.75}{%
	\begin{picture}(0,0)%
\includegraphics{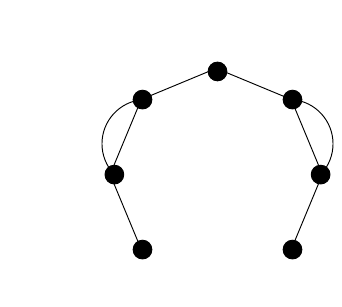}%
\end{picture}%
\setlength{\unitlength}{3947sp}%
\begingroup\makeatletter\ifx\SetFigFont\undefined%
\gdef\SetFigFont#1#2#3#4#5{%
  \reset@font\fontsize{#1}{#2pt}%
  \fontfamily{#3}\fontseries{#4}\fontshape{#5}%
  \selectfont}%
\fi\endgroup%
\begin{picture}(2730,2436)(4486,-1789)
\put(5101,-436){\makebox(0,0)[lb]{\smash{{\SetFigFont{12}{14.4}{\familydefault}{\mddefault}{\updefault}{\color[rgb]{0,0,0}$e_{\beta}'$}%
}}}}
\put(7201,-436){\makebox(0,0)[lb]{\smash{{\SetFigFont{12}{14.4}{\familydefault}{\mddefault}{\updefault}{\color[rgb]{0,0,0}$e_{\beta}$}%
}}}}
\put(6151,239){\makebox(0,0)[lb]{\smash{{\SetFigFont{12}{14.4}{\familydefault}{\mddefault}{\updefault}{\color[rgb]{0,0,0}$v_0$}%
}}}}
\put(5176,-136){\makebox(0,0)[lb]{\smash{{\SetFigFont{12}{14.4}{\familydefault}{\mddefault}{\updefault}{\color[rgb]{0,0,0}$v_{j-1}$}%
}}}}
\put(6976,-1411){\makebox(0,0)[lb]{\smash{{\SetFigFont{12}{14.4}{\familydefault}{\mddefault}{\updefault}{\color[rgb]{0,0,0}$v_3$}%
}}}}
\put(7126,-736){\makebox(0,0)[lb]{\smash{{\SetFigFont{12}{14.4}{\familydefault}{\mddefault}{\updefault}{\color[rgb]{0,0,0}$v_2$}%
}}}}
\put(6976,-136){\makebox(0,0)[lb]{\smash{{\SetFigFont{12}{14.4}{\familydefault}{\mddefault}{\updefault}{\color[rgb]{0,0,0}$v_1$}%
}}}}
\put(4951,-736){\makebox(0,0)[lb]{\smash{{\SetFigFont{12}{14.4}{\familydefault}{\mddefault}{\updefault}{\color[rgb]{0,0,0}$v_{j-2}$}%
}}}}
\put(5101,-1411){\makebox(0,0)[lb]{\smash{{\SetFigFont{12}{14.4}{\familydefault}{\mddefault}{\updefault}{\color[rgb]{0,0,0}$v_{j-3}$}%
}}}}
\put(4501,-1711){\makebox(0,0)[lb]{\smash{{\SetFigFont{12}{14.4}{\familydefault}{\mddefault}{\updefault}{\color[rgb]{0,0,0}$\beta \notin \bar{\mathcal C}_{j-3}$}%
}}}}
\put(7126,-1711){\makebox(0,0)[lb]{\smash{{\SetFigFont{12}{14.4}{\familydefault}{\mddefault}{\updefault}{\color[rgb]{0,0,0}$\beta \notin \bar{\mathcal C_3}$}%
}}}}
\put(6001,464){\makebox(0,0)[lb]{\smash{{\SetFigFont{12}{14.4}{\familydefault}{\mddefault}{\updefault}{\color[rgb]{0,0,0}$\beta \in \bar{\mathcal C_0}$}%
}}}}
\end{picture}%
}
	\caption{The graph $H_A$}
	\label{fig:ha}
\end{figure}

To prove this claim, first assume without loss of generality that $\beta\in \bar\fc_0$.  Since $\beta$ is not missing at $v_1$, $\beta$ appears on an edge, say $e_{\beta}$, between $v_1$ and $v_2$ for the same reason that $\alpha_0$ does.  Likewise, since $\beta$ is not missing at $v_{j-1}$, $\beta$ appears on an edge $e_{\beta}'$ between $v_{j-1}$ and $v_{j-2}$.  Finally, suppose $\beta$ appears between $v_1$ and $v_2$, and is missing at $v_3$ in $c_0$.  Then let $e_\beta$ be the edge between $v_1$ and $v_2$ with colour $\beta$ in $c_0$.  We construct a colouring $c'_0$ from $c_0$ by giving $e_2$ colour $\beta$ and giving $e_\beta$ colour $\alpha_1$ (i.e.\ we swap the colours of $e_\beta$ and $e_2$).  Thus $c'_0$ is a $k$-edge-colouring of $G-e_0$ in which $\beta$ is missing at both $v_0$ and $v_1$.  We can therefore extend $G-e_0$ to a $k$-edge-colouring of $G$.  Thus if $\beta$ is missing at $v_{3}$ or $v_{j-3}$ we can easily $k$-edge-colour $G$.  We therefore have at least two edges of $H_A$ coloured $\beta$ for every vertex of $C$ at which $\beta$ is missing, and we do not double-count edges (see Figure \ref{fig:ha}). This proves the claim, and the analogous claim for any colour in $A$ also holds.

\noindent Now, taking indices modulo $j$, we have 

\begin{equation}
\sum_{i=0}^{j-1}\mu_{H_A}(v_iv_{i+1}) = |E(H_A)| > 2\sum_{i=0}^{j-1}(k-d_G(v_i)).
\end{equation}
Therefore
\begin{equation}
\sum_{i=0}^{j-1}\left( 2d_G(v_i) +  \mu_{H_A}(v_{i}v_{i+1}) \right) > 2jk.
\end{equation}
Rewriting,
\begin{equation}
\sum_{i=0}^{j-1}\left( d_G(v_i) + \tfrac 12 \mu_{H_A}(v_{i+1}v_{i+2}) + d_G(v_{i+2}) + \mu_{H_A}(v_{i}v_{i+1}) \right) > 2jk.
\end{equation}

Therefore there exists some index $i$ for which 
\begin{equation}
d_G(v_i) + \tfrac 12 \mu_{H_A}(v_{i+1}v_{i+2}) + d_G(v_{i+2}) + \tfrac 12 \mu_{H_A}(v_iv_{i+1}) > 2k.
\end{equation}
Therefore by (\ref{eqn:2}),
\begin{eqnarray}
2k\geq 2\bar \gall &\geq& d_G(v_i) + \tfrac 12(d(v_{i+1} - \mu_{H_A}(v_{i}v_{i+1}) )  + d_G(v_{i+2}) + \tfrac 12(d(v_{i+1} - \mu_{H_A}(v_{i+1}v_{i+2}) ) \\
&\geq & d_G(v_i) + \tfrac 12 \mu_{G}(v_{i+1}v_{i+2}) + d_G(v_{i+2}) + \tfrac 12 \mu_{G}(v_iv_{i+1})\\
&>& 2k,
\end{eqnarray}
a contradiction.  So we can indeed find a $k$-edge-colouring of $G$. This contradicts the assertion that $G$ is a minimum counterexample and completes the proof.
\end{proof}

\section{Colouring quasi-line graphs with $\gall(G)$ colours}

In this section we extend our bound on the chromatic number to {\em quasi-line} graphs.  A graph is quasi-line if every vertex is {\em bisimplicial}, i.e.\ its neighbours can be covered by two cliques.  This class contains all circular interval graphs and all line graphs, and just like those two classes, the fractional and integer chromatic numbers agree asymptotically for quasi-line graphs \cite{kingr13}.

Quasi-line graphs are essentially constructed as a combination of line graphs and circular interval graphs.  We forgo a lengthy description of their structure and instead direct the unfamiliar reader to \cite{chudnovskykps12}, \cite{cssurvey}, and \cite{kingthesis}.  Here we present the bare minimum of what we need.

To proceed we must define {\em linear interval graphs}, which are also known as {\em proper interval graphs}\footnote{The divergence of terminology is an unfortunate consequence of two possible definitions: one in which intervals represent cliques, and one, the original, in which intervals represent vertices.} \cite{denghh96}.  A graph $G=(V,E)$ is a linear interval graph precisely if it has a {\em linear interval representation}.  A linear interval representation consists of a point on the real line for each vertex, and a set of intervals such that vertices $u$ and $v$ are adjacent in $G$ precisely if there is an interval containing both corresponding points on the real line.  If $X$ and $Y$ are specified cliques in $G$ consisting of the $|X|$ leftmost and $|Y|$ rightmost vertices (with respect to the real line) of $G$ respectively, we say that $X$ and $Y$ are {\em end-cliques} of $G$.  These cliques may be empty.  We now describe how we might isolate a linear interval graph within a quasi-line graph.

\begin{figure}[ht]
	\centering
	\scalebox{0.7}{
	\begin{picture}(0,0)%
\includegraphics{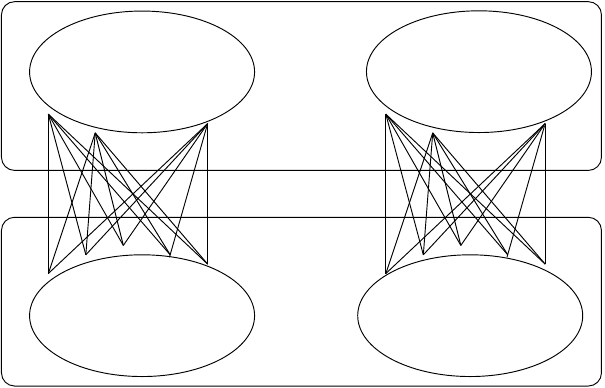}%
\end{picture}%
\setlength{\unitlength}{3947sp}%
\begingroup\makeatletter\ifx\SetFigFont\undefined%
\gdef\SetFigFont#1#2#3#4#5{%
  \reset@font\fontsize{#1}{#2pt}%
  \fontfamily{#3}\fontseries{#4}\fontshape{#5}%
  \selectfont}%
\fi\endgroup%
\begin{picture}(4824,3099)(4939,-4798)
\put(7126,-3886){\makebox(0,0)[lb]{\smash{{\SetFigFont{14}{16.8}{\familydefault}{\mddefault}{\updefault}{\color[rgb]{0,0,0}$G_1$}%
}}}}
\put(6001,-2986){\makebox(0,0)[lb]{\smash{{\SetFigFont{12}{14.4}{\familydefault}{\mddefault}{\updefault}{\color[rgb]{0,0,0}...}%
}}}}
\put(8701,-2986){\makebox(0,0)[lb]{\smash{{\SetFigFont{12}{14.4}{\familydefault}{\mddefault}{\updefault}{\color[rgb]{0,0,0}...}%
}}}}
\put(5776,-4336){\makebox(0,0)[lb]{\smash{{\SetFigFont{14}{16.8}{\familydefault}{\mddefault}{\updefault}{\color[rgb]{0,0,0}$X_1$}%
}}}}
\put(5776,-2311){\makebox(0,0)[lb]{\smash{{\SetFigFont{14}{16.8}{\familydefault}{\mddefault}{\updefault}{\color[rgb]{0,0,0}$X_2$}%
}}}}
\put(8551,-4336){\makebox(0,0)[lb]{\smash{{\SetFigFont{14}{16.8}{\familydefault}{\mddefault}{\updefault}{\color[rgb]{0,0,0}$Y_1$}%
}}}}
\put(8551,-2311){\makebox(0,0)[lb]{\smash{{\SetFigFont{14}{16.8}{\familydefault}{\mddefault}{\updefault}{\color[rgb]{0,0,0}$Y_2$}%
}}}}
\put(7126,-2086){\makebox(0,0)[lb]{\smash{{\SetFigFont{14}{16.8}{\familydefault}{\mddefault}{\updefault}{\color[rgb]{0,0,0}$G_2$}%
}}}}
\end{picture}%
}
	\caption{A canonical interval 2-join}
	\label{fig:ci2j}
\end{figure}
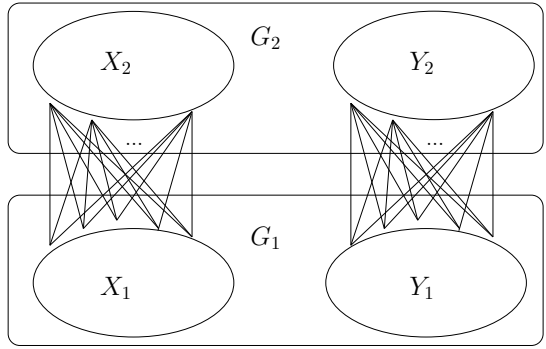

Given four cliques $X_1$, $Y_1$, $X_2$, and $Y_2$, we say that $((V_1,X_1,Y_1),\allowbreak (V_2,X_2,Y_2))$ is a {\em canonical interval 2-join} if it satisfies the following conditions (see Figure \ref{fig:ci2j}):
\begin{itemize}
\item $V(G)$ can be partitioned into nonempty $V_1$ and $V_2$ with $X_1\cup Y_1\subseteq V_1$ and $X_2\cup Y_2\subseteq V_2$ such that for $v_1\in V_1$ and $v_2\in V_2$, $v_1v_2$ is an edge precisely if $\{v_1,v_2\}$ is in $X_1\cup X_2$ or $Y_1\cup Y_2$.
\item $G|V_2$ is a linear interval graph with disjoint end-cliques $X_2$ and $Y_2$.
\end{itemize}

We now say that a quasi-line graph $G$ is a minimum counterexample if $\chi(G) > \gall(G)$, but no smaller quasi-line graph has the same property.  Note that any induced subgraph of a quasi-line graph is quasi-line.  Theorems \ref{thm:cig} and \ref{thm:line}, combined with well-known structural results (e.g.\ Theorem 16 and the discussion in Sections 3.3-3.5 of \cite{chudnovskykps12}), imply:

\begin{proposition}\label{lem:mceql}
If $G$ is a minimum counterexample then $G$ admits a canonical interval 2-join.
\end{proposition}

The main result of this section is:

\begin{theorem}\label{thm:ql}
Let $G$ be a quasi-line graph. Then $\chi(G) \leq \gall(G)$.
\end{theorem}

\noindent{\bf Remark:} As with Theorem \ref{thm:line}, here we implicitly prove that the algorithm from \cite{chudnovskykps12} uses at most $\gall(G)$ colours.  This gives us a time complexity bound of $O(n^2m^2)$, which we believe can be improved to $O(m^2)$.\\

To prove Theorem \ref{thm:ql} it only remains to prove that a minimum counterexample cannot contain a canonical interval 2-join.  Given a canonical interval 2-join $((V_1,X_1, Y_1),\allowbreak (V_2,X_2, Y_2))$ in $G$ with an appropriate partitioning $V_1$ and $V_2$, let $G_1$ denote $G|V_1$, let $G_2$ denote $G|V_2$ and let $H_2$ denote $G|(V_2 \cup X_1 \cup Y_1)$.  For $v \in H_2$ we define $\omega'(v)$ as the size of the largest clique in $H_2$ containing $v$ and not intersecting both $X_1 \setminus Y_1$ and $Y_1 \setminus X_1$.  For $uv\in E(H_2)$ we define $\gall^j(uv)$ as $\lceil \tfrac 14 (d_G(u) + d_G(v)+2+\omega'(u) +\omega'(v))\rceil$, and we define $\gall^j(H_2)$ as $\max_{uv\in E(H_2)} \gall^j(uv)$ (here the superscript $j$ denotes {\em join}).  Observe that $\gall^j(H_2) \leq \gall(G)$.  If $v \in X_1 \cup Y_1$, then $\omega'(v)$ is $|X_1|+|X_2|$, $|Y_1|+|Y_2|$, or $|X_1\cap Y_1|+ \omega(G|(X_2 \cup Y_2))$.  So rather than bounding $\chi$ by $\gall$, we bound $\chi$ by a refinement of $\gall$ derived from our decomposition.

\begin{lemma}\label{lem:quasilinemce2}
Let $G$ be a minimum counterexample admitting a canonical interval 2-join $((V_1,X_1, Y_1),\allowbreak (V_2,X_2, Y_2))$.  Then given a proper $l$-colouring of $G_1$ for any $l \geq \gall^j(H_2)$, we can find a proper $l$-colouring of $G$.
\end{lemma}

Since $\chi(G_1) \leq \gall(G_1)\leq \gall(G)$ and $\gall^j(H_2) \leq \gall(G)$, this lemma immediately implies Theorem \ref{thm:ql}.
Also since a minimum counterexample cannot contain a clique cutset (this is a straightforward observation since no graph with a clique cutset is vertex-critical), all four cliques $X_1$, $Y_1$, $X_2$, and $Y_2$ must be nonempty.

\begin{proof}
We proceed by induction on $l$, observing that the case $l = 1$ is trivial.  We begin by modifying the colouring so that the number $k$ of colours used in both $X_1$ and $Y_1$ in the $l$-colouring of $G_1$ is maximal.  That is, if a vertex $v \in X_1$ gets a colour that does not appear in $Y_1$, then every colour appearing in $Y_1$ appears in $N(v)$.  If $l$ exceeds $\gall^j(H_2)$ we can just remove a colour class in $G_1$ and apply induction on what remains.  Thus we can assume that $l =\gall^j(H_2)$ and so if we apply induction we must remove a stable set whose removal lowers both $l$ and $\gall^j(H_2)$.

We use case analysis; when considering a case we may assume no previous case applies.  In some cases we extend the colouring of $G_1$ to an $l$-colouring of $G$ in one step.  In other cases we remove a colour class in $G_1$ together with vertices in $G_2$ such that everything we remove is a stable set, and when we remove it we reduce $\gall^j(uv)$ for every $uv \in E(H_2)$; after doing this we apply induction on $l$.  Notice that if $X_1 \cap Y_1 \neq \emptyset$ and there are edges between $X_2$ and $Y_2$ we may have a large clique in $H_2$ which contains some but not all of $X_1$ and some but not all of $Y_1$; this is a subtlety that we deal with in every applicable case.

\begin{itemize}
\item[Case 1.]$Y_1 \subseteq X_1$.

Since $G$ cannot contain a clique cutset, $H_2 = G$ and furthermore $H_2$ is a circular interval graph, contradicting the assumption that $G$ is a minimum counterexample.

\item[Case 2.]$k=0$ and $l > |X_1|+|Y_1|$.

Here $X_1$ and $Y_1$ are disjoint since $k=0$.  Take a stable set $S$ greedily from left to right in $G_2$.  By this we mean that we start with $S=\{v_1\}$ (the leftmost vertex of $X_2$) and we move along the vertices of $G_2$ in linear order, adding a vertex to $S$ whenever doing so will leave $S$ a stable set.  So $S$ hits $X_2$.  If it hits $Y_2$, remove $S$ along with a colour class in $G_1$ not intersecting $X_1\cup Y_1$; these vertices together make a stable set.  If $v\in G_2$ it is easy to see that $\omega'(v)$ will drop: $S$ intersects every maximal clique containing $v$.  If $v\in X_1 \cup Y_1$ then since $X_1$ and $Y_1$ are disjoint, $\omega'(v)$ is either $|X_1|+|X_2|$ or $|Y_1|+|Y_2|$; in either case $\omega'(v)$ drops.  Therefore since $S$ is maximal in $H_2$, $\gall^j(uv)$ drops for each edge $uv\in E(H_2)$ when $S$ and the colour class are removed.  Therefore $\gall^j(H_2)$ and $l$ drop, and we can proceed by induction.

If $S$ does not hit $Y_2$ we remove $S$ along with a colour class from $G_1$ that hits $Y_1$ (and therefore not $X_1$).  Since $S\cap Y_2 = \emptyset$ the vertices together make a stable set.  Using the same argument as before we can see that removing these vertices drops both $l$ and $\gall^j(H_2)$, so we can proceed by induction.

\item[Case 3.]$k=0$ and $l = |X_1|+|Y_1|$.

Again, $X_1$ and $Y_1$ are disjoint.  Since $G$ cannot contain a clique cutset, $G_2$ is connected.  
Therefore every vertex in $X_2$ must have a neighbour outside $X_1$.  
Consequently $\gall^j(H_2) > |X_1\cup X_2| \geq |X_1|+1$.  Since $l \geq \gall^j(H_2)$, this implies that $|Y_1|>1$.  
The symmetric argument tells us that $|X_1|>1$.

By maximality of $k$, every vertex in $X_1 \cup Y_1$ has at least $l-1$ neighbours in $G_1$. Since $l=|X_1|+|Y_1|$ and $\gall^j(X_1\cup Y_1) \leq l$, we know that $\omega'(X_1) \leq |X_1|+|Y_1|-|X_2|$ and $\omega'(Y_1) \leq |X_1|+|Y_1|-|Y_2|$.  Thus $|Y_1| \geq 2|X_2|$ and similarly $|X_1| \geq 2|Y_2|$.  For the remainder of this case we assume without loss of generality that $|Y_2| \leq |X_2|$.

We first attempt to $l$-colour $H_2 - Y_1$, which we denote by $H_3$,
such that every colour in $Y_2$ appears in $X_1$ -- this is clearly
sufficient to prove the lemma since we can permute the colour classes
and paste this colouring onto the colouring of $G_1$ to get a proper
$l$-colouring of $G$.  If $\omega(H_3) \leq l-|Y_2|$ then this is
easy:  since $H_3$ is a linear interval graph we can $\omega(H_3)$-colour 
the vertices of $H_3$, then use
$|Y_2|$ new colours to recolour $Y_2$ and $|Y_2|$ vertices of $X_1$.
This is possible since $Y_2$ and $X_1$ have no edges between them.
Defining $b$ as $l - \omega(H_3)$, we can now assume that $b < |Y_2|$.

It now suffices to find an $\omega(H_3)$-colouring of $H_3$ such that
at most $b$ colours appear in $Y_2$ but not $X_1$. 
This is because if we take such a colouring and permute the colours so that they agree with
our $l$-colouring of $G_1$ on $X_1$, we can use the colours which don't yet appear on $H_3$
to recolour $b$ vertices in $Y_2$ to obtain a proper colouring.
There is some
clique $C = \{v_i, \ldots, v_{i+\omega(H_3)-1}\}$ in $H_3$; this
clique does not intersect $X_1$ because $|X_1 \cup X_2| \leq l-|X_2| \leq
l - |Y_2|< l-b = \omega(H_3)$, where the first inequality follows from $l = |X_1| + |Y_1|$ and $|Y_1| \geq 2|X_2|$. 
Since $\gall^j(v_iv_{i+1}) \leq l$, it is clear that
either $v_i$ or $v_{i+1}$ has at most $2b$ neighbours outside $C$. Let
$v_{i'}\in \{v_i,v_{i+1}\}$ be the vertex with this property. Since $b
< |Y_2| \leq \frac 12 |X_1|$ we can be assured that $v_{i'} \notin
X_2$.  
Since $\omega(H_3)\geq |X_1| + |X_2| > |Y_2|$, we deduce $v_{i'} \notin Y_2$.

We now colour $H_3$ greedily from left to right, modulo $\omega(H_3)$.  
If at most $b$ colours appear in $Y_2$ but not $X_1$ then we are done, otherwise we will ``roll back'' the colouring, starting at $v_{i'}$.  That is, for every $p \geq i'$, we modify the colouring of $H_3$ by giving $v_p$ the colour after the one that it currently has, modulo $\omega(H_3)$.  Since $v_{i'}$ has at most $2b$ neighbours behind it, we can roll back the colouring at least $\omega(H_3)-2b-1$ times for a total of $\omega(H_3)-2b$ proper colourings of $H_3$.

Since $v_{i'} \notin Y_2$ the colours on $Y_2$ will appear in order modulo $\omega(H_3)$ in all of the rolled back colourings.  
The colours on $X_1$ will also be in order.
Thus the colouring of $Y_2$ is one of at most $\omega(H_3)$ possibilities in each rolled back colouring, 
and in $2b+1$ of them there are at most $b$ colours appearing in $Y_2$ but not $X_1$.  
It follows that one of the rolled back colourings of $H_3$ will be acceptable.

Henceforth we drop the assumption that $|X_2|\geq |Y_2|$, and assume without loss of generality that $|X_1| \geq |Y_1|$.

\item[Case 4.]$0 < k < |X_1|$.

Take a stable set $S$ in $G_2 - X_2$ greedily from left to right.  If $S$ hits $Y_2$, we remove $S$ from $G$, along with a colour class from $G_1$ intersecting $X_1$ but not $Y_1$.  Otherwise, we remove $S$ along with a colour class from $G_1$ intersecting both $X_1$ and $Y_1$.  In either case it is a simple matter to confirm that $\gall^j(uv)$ drops for every $uv \in E(H_2)$ as we did in Case 2.  We proceed by induction.

\item[Case 5.]$k=|Y_1|=|X_1|=1$.

In this case $|X_1|=k=1$.  If $G_2$ is not connected then $X_1$ and $Y_1$ are both clique cutsets and we can proceed as in Case 1.  If $G_2$ is connected and contains an $l$-clique, then there is some $v \in V_2$ of degree at least $l$ in the $l$-clique.  Thus $\gall^j(H_2) > l$, contradicting our assumption that $l \geq \gall^j(H_2)$.  So $\omega(G_2)<l$.  We can $\omega(G_2)$-colour $G_2$ in linear time using only colours not appearing in $X_1 \cup Y_1$, thus extending the $l$-colouring of $G_1$ to a proper $l$-colouring of $G$.

\item[Case 6.]$k=|Y_1|=|X_1|> 1$.

Suppose that $k$ is not minimal.  That is, suppose there is a vertex $v \in X_1 \cup Y_1$ whose closed neighbourhood does not contain all $l$ colours in the colouring of $G_1$.  Then we can change the colour of $v$ and apply the argument of Case 4.  So assume $k$ is minimal.

Therefore every vertex in $X_1$ has degree at least $l+|X_2|-1$.  Since $X_1\cup X_2$ is a clique and $X_1$ contains an edge, $l \geq \gall^j(H_2) \geq \frac 12 (l+|X_2|+|X_1|+|X_2|)$, so $2|X_2|\leq l-k$.  Similarly, $2|Y_2|\leq l-k$, so $|X_2|+|Y_2|\leq l-k$.   Since $X_1$ and $Y_1$ contain the same $k$ colours, there are $l-k$ colours not appearing in $X_1\cup Y_1$ in the $l$-colouring of $G_1$, so we can $\omega(G_2)$-colour $G_2$, then permute the colour classes so that no colour appears in both $X_1\cup Y_1$ and $X_2 \cup Y_2$.  Thus we can extend the $l$-colouring of $G_1$ to an $l$-colouring of $G$.
\end{itemize}

These cases cover every possibility, so the lemma is proved.
\end{proof}
This completes the proof of Theorem \ref{thm:ql}.

\section{Conclusion}\label{sec:conc}

The local version of Reed's Conjecture proposes that Reed's two requirements for high chromatic number, namely high degree and high clique number, must occur in the same part of the graph.  The superlocal version proposes that this must occur at least twice in the same part of the graph.  We believe that this requirement can be pushed further, at least in the fractional setting.  Let $\mathcal C(G)$ be the set of maximal cliques in a graph $G$.

\begin{conjecture}\label{con:clique}
Every graph $G$ satisfies
$$\chi_f(G) \leq \max_{C\in \mathcal C(G)}\frac{1}{|C|}\sum_{v\in C}  \gamma_\ell'(v).$$
\end{conjecture}

We cannot hope to take the maximum average over a closed neighbourhood rather than the maximum average over a maximal clique.  To see this, take a clique $C$ of size $k$ and attach $k$ pendant vertices to every vertex of $C$.  Each $v$ in $C$ has $\gamma_\ell'(v) = \frac 32 k$, and each $u \notin C$ has $\gamma_\ell'(u) = 2$.  Therefore for any $v$,
$$\frac{1}{d(v)+1}\sum_{u\in \tilde N(v)}  \gamma_\ell'(u) = \tfrac 34k + 1.$$
For $k>4$, this is less than the fractional chromatic number, i.e.\ $k$.  However, we would like to know if the condition holds when it is no longer possible to lower the bound by adding vertices:

\begin{question}
Does every graph $G$ satisfy
$$\chi_f(G) \leq \max_{H \subseteq G}\ \max_{v\in V(H)}\frac{1}{d_H(v)+1}\sum_{u\in \tilde N_H(v)}  {\gamma'_\ell}_H(u)?$$
\end{question}

If true, this would be very interesting, since it would require a different colouring method than the one used in Section \ref{sec:frac}.  To see this, consider the tree on six vertices, four of which have degree 1 and two of which have degree 3.  In this case the fractional colouring process described in Section \ref{sec:frac} gives a fractional 3-colouring, whereas the bound in question is $5/2$.

On the subject of integer colouring, proving Conjecture \ref{con:3} for claw-free graphs does not seem easier than proving Conjecture \ref{con:2} for claw-free graphs.  However, proofs of the local version seem easy to extend to the superlocal version.  In particular, we believe it should be easy to prove Conjecture \ref{con:3} for claw-free graphs with $\alpha \leq 3$, following the proof in \cite{kingthesis}.

\section{Acknowledgements}

We thank the referee for a careful and helpful review, and the editors for their time and contribution to the journal.

\end{document}